\documentclass{llncs}

\usepackage{amsfonts}%
\usepackage{amssymb}%
\usepackage{graphicx}%
\usepackage{amsmath}%
\usepackage{setspace}
\usepackage[usenames]{color}
\usepackage{amssymb,amsmath}
\usepackage{multirow,url}
\usepackage{algorithmic}

\newcommand{\comment}[1]{}

\ifodd 0
\newcommand{\rev}[1]{{\color{blue}#1}} %revise of the text
\newcommand{\com}[1]{\textbf{\color{red}(COMMENT: #1)}} %comment of the text
\newcommand{\clar}[1]{\textbf{\color{green}(NEED CLARIFICATION: #1)}}
\else
\newcommand{\rev}[1]{#1}
\newcommand{\com}[1]{}
\newcommand{\clar}[1]{}
\fi
\begin{document}
\title{Performance and Convergence of Multi-user \\ \rev{Online} Learning}
\author{Cem Tekin \and Mingyan Liu}
\institute{Department of Electrical Engineering and Computer Science\\
University of Michigan, Ann Arbor, Michigan, 48109-2122\\ \email{\{cmtkn, mingyan\}@umich.edu}}
\maketitle

\begin{abstract}  
We study the problem of allocating multiple users to a set of wireless channels in a decentralized manner 
when the channel qualities are time-varying and unknown to the users, and accessing the same channel by multiple users
leads to reduced quality due to interference.  In such a setting the users not only need to learn
the inherent channel quality and at the same time the best allocations of users to channels so as to maximize the social welfare. 
Assuming that the users adopt a certain online learning algorithm, we investigate under what conditions the 
socially optimal allocation is achievable.  In particular we examine the effect of different levels of knowledge the users 
may have and the amount of communications and cooperation.  
The general conclusion is that when the cooperation of users decreases and the
uncertainty about channel payoffs increases it becomes harder to
achieve the socially optimal allocation. 
\end{abstract}

\section{Introduction} \label{sec:intro}

In this paper we study the dynamic spectrum access and spectrum
sharing problem in a learning context.  Specifically, we consider
a set of $N$ common channels shared by a set of $M$ users.  A
channel has time varying rate $r(t)$, and its statistics are not
completely known by the users.  Thus each user needs to employ
some type of learning to figure out which channels are of better
quality, e.g., in terms of their average achievable rates. At the
same time, simultaneous use of the same channel by multiple users
will result in reduced rate due to interference or collision.
The precise form of this performance degradation may or may not be known to the user.  Thus the users also need to use learning to avoid excess interference or congestion.  Furthermore, each user may have private information that is not shared, e.g., users may perceive channel quality differently due to difference in location as well as individual modulation/coding schemes.

Without a central agent, and in the presence of information
decentralization described above, we are interested in the
following questions: (1) for a given common learning algorithm,
does the multiuser learning process converge, and
(2) if it does, what is the quality of the equilibrium point with
respect to a globally optimal spectrum allocation scheme, one that
could be computed for a global objective function with full
knowledge of channel statistics as well as \rev{the users'} private
information.

A few recent studies have addressed these questions in some
special cases. For instance, in \cite{anandkumar} it was shown
that learning using a sample-mean based index policy leads to a
socially optimal (sum of individual utilities) allocation when
channels evolve as iid processes and colliding players get zero
reward provided that this optimal allocation is such that each user
occupies one of the $M$ best channels (in terms of average rates).
This precludes the possibility that not all users may have the
same set of $M$ best channels, and that in some cases the best
option is for multiple users to share a common channel, e.g., when
$N<M$.

In this study we investigate under what conditions the socially
optimal allocation is achievable by considering different levels of
communication (or cooperation) allowed among users, and different
levels of uncertainty on the channel statistics. The general
conclusion, as intuition would suggest, is that when the
cooperation of users increases and the channel uncertainty
decreases it becomes easier to achieve the socially optimal
welfare. Specifically, we assume that the rate (or reward) user $i$ gets
from channel $j$ at time $t$ is of the form $r_j(t) g_j(n_j(t))$
where $r_j(t)$ is the rate of channel $j$ at time $t$, $n_j(t)$ is
the number of users using channel $j$ at time $t$, and $g_j$ is the
user independent interference function (IF) for channel
$j$. This model is richer than the previously used models
\cite{anandkumar,kleinberg1,liu2} since $r_j(t)$ can represent
environmental effects such as fading or primary user activity,
while $g_j$ captures interactions between users. We consider the
following three cases. 

In the first case (C1), each channel evolves as an iid random
process in time, the users do not know the channel statistics,
nor the form of the interference, nor the total number of users
present in the system, and no direct communication is allowed
among users. A user can measure the overall rate it gets from
using a channel but cannot tell how much of it is due to the
dynamically changing channel quality (i.e., what it would get if
it were the only user) vs. interference from other users.  In this
case, we show that if all users follow the Exp3 algorithm
\cite{auer2} then the channel allocation converges to a set of
pure Nash equilibria (PNE) of a congestion game defined by the
IFs and mean channel rates.  In this case a socially optimal
allocation cannot be ensured, as the set of PNE are of different
quality, and in some cases the socially optimal allocation may not
be a PNE.

In the second case (C2), each channel again evolves as an iid
random process in time, whose statistics are unknown to the user.
However, the users now know the total number of users in the
system, as well as the fact that the quantitative impact of
interference is common to all users (i.e., user independent),
though the actual form of the interference function is unknown.
In other words the rate of channel $j$ at time $t$ is perceived by
user $i$ as $h_j(t,n_j(t))$ so user $i$ cannot distinguish between
components $r_j(t)$ and $g_j(n_j(t))$. 
Furthermore, users are now allowed minimal amount of communication
when they happen to be in the same channel, specifically to find
out the total number of simultaneous users of that channel. In this case we present a
sample-mean based randomized learning policy that achieves
socially optimal allocation as time goes to infinity, with a
sub-linear regret over the time horizon with respect to the
socially optimal allocation.

In the third case (C3), as in case (C2) the users know the total
number of users in the system, as well as the fact that the
IF is user independent and decreasing without knowing the
actual form of the IF.  However, the channels
are assumed to have constant, albeit unknown, rates.  We show that
even without any communication among users, there is a randomized
learning algorithm that achieves the socially optimal allocation
in finite time.

It's worth pointing out that in the settings outlined above, the users are 
{\em non-strategic}, i.e., each user simply follow a pre-set learning rule 
rather than playing a game.  In this context it is reasonable to introduce minimal amount of 
communication among users and assume they may cooperate. 
It is possible that even in this case the users may not know their IF 
but only the total rate they get for lack of better detecting capabilities (e.g., they may only 
be able to detect the total received SNR as a result of channel rate and user interference). 

Online learning by a single user was studied by \cite{agrawal1,anantharam2,auer,lai1}, in which sample-mean based index policies were shown to achieve 
logarithmic regret with respect to the best single-action policy without a priori 
knowledge of the statistics, and are order-optimal, when the rewards are given by 
an iid process. In \cite{anantharam1,tekin,tekin2} Markovian rewards are considered, 
with \cite{tekin2} focusing on {\em restless} reward processes, where a process continues
to evolve according to a Markov chain regardless of the users' actions.  In all these 
studies learning algorithms were developed to achieve logarithmic regret. 
Multi-user learning with iid reward processes have been studied in a dynamic spectrum context
by \cite{anandkumar,krishnamachari,liu2}, with a combinatorial structure adopted in \cite{krishnamachari}, 
and with collision and random access models in \cite{anandkumar,liu2}. In \cite{gaurav2010}, convergence of multi-user learning with Exp3 algorithm to pure Nash equilibrium is investigated under the collision and fair sharing models. In the collision model, when there is more
than one user on a channel all get zero reward, whereas in the
random access model one of them, selected randomly, gets all the reward
while others get zero reward. In the fair sharing model, a user's utility is inversely proportional to the number of users who are on the same channel with the user. Note that these models do not
capture more sophisticated communication schemes where the rate a user gets is 
a function of the received SNR of the form 
$g_j(n)= f_j(\frac{P_t}{N_0+(n-1) P_t})=$ where $P_t$ is the nominal transmit power of all users
and $N_0$ the noise.  
Moreover, in the above studies the socially optimal allocation is a rather simple one: 
it is the orthogonal allocation of users to the first $M$ channels with the highest mean rewards. 
By contrast, we model a more general interference relationship among users, in which an allocation with users sharing the same channel may be the socially optimal one.
%between users.
The socially optimal allocations is not trivial in this case and
additional mechanisms may be needed for the learning algorithms to converge.

All of the above mentioned work assumes some level of
communication between the users either at the beginning or during
the learning. If we assume no communication between the users, achieving the socially optimal allocation seems very challenging
in general. Then one may ask if it is possible to achieve some
kind of equilibrium allocation. Kleinberg et. al.
\cite{kleinberg1} showed that it is possible for the case when the
channel rates are constant and the users do not know the IFs. They
show that when the users use {\em aggregate monotonic selection
dynamics}, a variant of Hedge algorithm \cite{freund1}, the
allocation converges to {\em weakly stable equilibria} which is a
subset of Nash equilibria (NE) of the congestion game defined by
the IFs. They show that for almost all congestion games weakly
stable equilibria is the same as PNE.

Other than the work described above \cite{ahmad2} considers {\em
spatial congestion games}, a generalization of congestion games
and gives conditions under which there exists a PNE and
best-response play converges to PNE. A mechanism design approach
for socially optimal power allocation when users are strategic is
considered in \cite{kakhbod1}.

The organization of the remainder of this paper is as follows. In
Sect. \ref{sec:problem} we present the notations and definitions
that will be used throughout the paper. In Sects.
\ref{sec:case1}, \ref{sec:case2}, \ref{sec:case3} we analyze the
cases stated in (C1), (C2), (C3) and derive the results
respectively. Conclusion and future research is given in Sect. \ref{sec:conc}.

\section{Preliminaries} \label{sec:problem}

Denote the set of users by $\mathcal{M}=\{1,2,\ldots,M\}$, and the set of
channels $\mathcal{N}=\{1,2,\ldots,N\}$. Time is slotted and indexed by 
$t=1,2,\ldots$ and a user can select a single channel at each time
step $t$. Without loss of generality let $r_j(t) \in [0,1]$ be the
rate of channel $j$ at time $t$ such that $\{r_j(t)\}_{t=1,2,\ldots}$ is generated by an iid
process with support $[0,1]$ and mean $\mu_j \in [0,1]$. Let $g_j: \mathbb{N} \rightarrow [0,1]$
be the interference function (IF) on channel $j$\comment{where $g_j(n)$
is a positive decreasing function for $n \geq 1$ with $g_j(0)=0$}
where $g_j(n)$ represents the interference when there are $n$
users on channel $j$. We express the rate of channel $j$ seen by a user as $h_j(t)=r_j(t)g_j(n_j(t))$ when a user does not know the total number of users $n_j(t)$ using channel $j$ at time $t$ as in cases (C1) and (C3). When a user knows $n_j(t)$, we express the rate of channel $j$ at time $t$ as $h_{j, n_j(t)}(t) = r_j(t)g_j(n_j(t))$ as in case (C2). Let $\mathcal{S}_i={\cal N}$
be the set of feasible actions of user $i$ and $\sigma_i \in
\mathcal{S}_i$ be the action, i.e., channel selected by user $i$.
Let $\mathcal{S}= \mathcal{S}_1 \times \mathcal{S}_2
\times \ldots \times \mathcal{S}_N = {\cal N}^M$
be the set of feasible action profiles and $\sigma=\{\sigma_1,
\sigma_2, \ldots, \sigma_M \} \in \mathcal{S} $ be the action
profile of the users. Throughout the discussion we assume that the action of player $i$ at time $t$, i.e., $\sigma_i^{\pi_i} (t)$ is determined by the policy $\pi_i$. When $\pi_i$ is deterministic,
$\pi_i(t)$ is in general a function from all past
observations and decisions of user $i$ to the set of actions
$\mathcal{S}_i$. When $\pi_i$ is randomized, $\pi_i(t)$ generates
a probability distribution over the set of actions $\mathcal{S}_i$
according to all past observations and decisions of user $i$ from
which the action at time $t$ is sampled. Since the dependence of actions to the policy is trivial we use $\sigma_i(t)$ to denote the action of user $i$ at time $t$, dropping the superscript $\pi_i$.

Let $K_j(\sigma)$ be the set
of users on channel $j$ when the action profile is $\sigma$. Let $\mathcal{A}^* =
\arg\max_{\sigma \in \mathcal{S}} \sum_{i=1}^M \mu_{\sigma_i}
g_{\sigma_i}(K_{\sigma_i}(\sigma)) = \arg\max_{\sigma \in
\mathcal{S}} \sum_{j=1}^N \mu_{j} K_{j}(\sigma)
g_{j}(K_{j}(\sigma))$ be the set of socially optimal allocations
and denote by $\sigma^*$ any action profile that is in the set
$\mathcal{A}^*$. Let $v^*$ denote the socially optimal welfare,
i.e., $v^*= \sum_{i=1}^M \mu_{\sigma^*_i}
g_{\sigma^*_i}(K_{\sigma^*_i}(\sigma^*))$ and $v^*_j$ denote the payoff a user gets from channel $j$ under the socially optimal allocation, i.e., $v^*_j = \mu_j g_{j}(K_{j}(\sigma^*))$ if $K_{j}(\sigma^*) \neq 0$. %and $v^*_j=1$ otherwise.
Note that any permutation
of actions in $\sigma^*$ is also a socially optimal allocation
since IFs are user-independent.

For any policy $\pi$, the regret at time $n$ is
\begin{eqnarray}
R(n) = n v^* -
E\left[\sum_{t=1}^n \sum_{i=1}^{M} r_{\sigma_i(t)}(t)
g_{\sigma_i(t)}(K_{\sigma_i(t)}(\sigma(t)))\right], \nonumber
\end{eqnarray}
where expectation is taken with respect to the random nature of
the rates and the randomization of the policy. Note that for a
deterministic policy expectation is only taken with respect to the
random nature of the rates. For any randomized policy $\pi_i$, let
$p_i(t)=(p_{i1}(t), p_{i2}(t), \ldots, p_{iN}(t))$ be the mixed
strategy of user $i$ at time $t$, i.e., a probability distribution
on $\{1,2,\ldots,N\}$. For a profile of policies $\pi=[\pi_1,
\pi_2, \ldots, \pi_M]$ for the users let $p(t)=(p_1(t)^T,
p_2(t)^T, \ldots p_M(t)^T)^T$ be the profile of mixed strategies at
time $t$, where $p_i(t)^T$ is the transpose of $p_i(t)$. Then
$\sigma_i(t)$ is the action sampled from the probability
distribution $p_i(t)$. The dependence of $p$ to $\pi$ is trivial
and not shown in the notation.

\section{Allocations Achievable with Exp3 Algorithm (Case 1)} \label{sec:case1}

We start by defining a congestion game. A congestion game \cite{monderer1,rosenthal1} is given by the
tuple $({\cal M}, {\cal N}, (\Sigma_i)_{i\in{\cal M}},
(h_j)_{j\in {\cal N}})$, where ${\cal M}$
denotes a set of players (users), ${\cal N}$ a set of resources (channels), $\Sigma_i\subset 2^{\cal N}$ the strategy space of player $i$, and $h_j: \mathbb{N}\rightarrow\mathbb{R}$ a payoff
function associated with resource $j$, which is a function of the number of players using that resource. It is well known that a congestion game has a potential function and the local maxima of the potential function corresponds to PNE, and every sequence of asynchronous improvement steps is finite and converges to PNE.

In this section we relate the strategy update rule of Exp3
\cite{auer2} under assumptions (C1) to a congestion game.
Exp3 as given in Fig. \ref{fig:exp3} is a randomized algorithm
consisting of an exploration parameter $\gamma$ and weights
$w_{ij}$ that depend exponentially on the past observations where
$i$ denotes the user and $j$ denotes the channel. Each user runs
Exp3 independently but we explicitly note the user dependence
because a user's action affects other users' updates.

\begin{figure}[htb]
\fbox {
\begin{minipage}{\columnwidth}
{Exp3 (for user $i$)}
\begin{algorithmic}[1]
\STATE {Initialize: $\gamma \in (0,1)$, $w_{ij}(t)=1, \forall j \in {\cal N}$, $t=1$}
\WHILE{$t>0$}
\STATE{\begin{eqnarray*}
\hspace{-2.4in} p_{ij}(t)=(1-\gamma)\frac{w_{ij}(t)}{\sum_{l=1}^N w_{il}(t)}+\frac{\gamma}{N}
\end{eqnarray*}
}
\STATE{Sample $\sigma_i(t)$ from the distribution on $p_i(t)=[p_{i1}(t), p_{i2}(t), \ldots, p_{iN}(t)]$}
\STATE{Play channel $\sigma_i(t)$ and receive reward $h_{\sigma_i(t)}(t)$}
\FOR{$j=1,2,\ldots,N$}
\IF{$j=\sigma_i(t)$}
\STATE{Set $w_{ij}(t+1)=w_{ij}(t) \exp\left({\frac{\gamma h_{\sigma_i(t)}(t)}{p_{ij}(t) N}}\right)$}
\ELSE
\STATE{Set $w_{ij}(t+1)=w_{ij}(t)$}
\ENDIF
\ENDFOR
\STATE{$t=t+1$}
\ENDWHILE
\end{algorithmic}
\end{minipage} }
\caption{pseudocode of Exp3} \label{fig:exp3}
\end{figure}

At any time step before the channel rate and user actions are
drawn from the corresponding distributions, let $R_j$ denote the
random variable corresponding to the reward of the $j$th channel.
Let $G_{ij}=g_j(1+K'_j(i))$ be the random variable representing the
payoff user $i$ gets from channel $j$ where $K'_j(i)$ is the random
variable representing the number of users on channel $j$ other
than user $i$. Let $U_{ij}=R_j G_{ij}$ and $\bar{u}_{ij} = E_j [
E_{-i}[U_{ij}]]$ be the expected payoff to user $i$ by using
channel $j$ where $E_{-i}$ represents the expectation taken with
respect to the randomization of players other than $i$, $E_j$
represents the expectation taken with respect to the randomization
of the rate of channel $j$. Since the channel rate is independent
of users' actions $\bar{u}_{ij} = \mu_j \bar{g}_{ij}$ where
$\bar{g}_{ij} = E_{-i}[G_{ij}]$.

\begin{lemma} \label{lemma:exp1}
Under (C1) when all players use Exp3, the derivative of the continuous-time limit of Exp3 is the replicator equation given by
\begin{eqnarray*}
\xi_{ij} = \frac{1}{N} (\mu_j p_{ij}) \sum_{l=1}^N p_{il} (\bar{g}_{ij}-\bar{g}_{il}) ~. 
\end{eqnarray*}
\end{lemma}
\begin{proof}
Note that
\begin{eqnarray}
(1-\gamma)w_{ij}(t)= \sum_{l=1}^N w_{il}(t) \left(p_{ij}(t)-\frac{\gamma}{N}\right) ~. \label{eqn:exp1}
\end{eqnarray}
We consider the effect of \rev{user $i$'s} action $\sigma_i(t)$ on his probability update on channel $j$.  \rev{We have two cases:} $\sigma_i(t)=j$ \rev{and} $\sigma_i(t) \neq j$. Let $A_{i,j}^{\gamma,t} = \exp\left( \frac{\gamma U_{ij}(t)}{p_{ij}(t)N} \right)$.

Consider the case $\sigma_i(t)=j$.
\begin{eqnarray}
p_{ij}(t+1)= \frac{(1-\gamma) w_{ij}(t) A_{i,j}^{\gamma,t}} {\sum_{l=1}^N w_{il}(t) + w_{ij}(t) \left( A_{i,j}^{\gamma,t}-1 \right)} + \frac{\gamma}{N} ~. \label{eqn:exp2}
\end{eqnarray}
Substituting (\ref{eqn:exp1}) into (\ref{eqn:exp2})
\begin{eqnarray}
p_{ij}(t+1)&=& \frac{ \sum_{l=1}^N w_{il}(t) \left(p_{ij}(t)-\frac{\gamma}{N}\right) A_{i,j}^{\gamma,t}} { \sum_{l=1}^N w_{il}(t) \left( 1 + \frac{p_{ij}(t)-\frac{\gamma}{N}}{1-\gamma} \left( A_{i,j}^{\gamma,t}-1 \right) \right)} + \frac{\gamma}{N} \nonumber \\
&=& \frac{\left(p_{ij}(t)-\frac{\gamma}{N}\right) A_{i,j}^{\gamma,t}}{1 + \frac{p_{ij}(t)-\frac{\gamma}{N}}{1-\gamma} \left( A_{i,j}^{\gamma,t}-1 \right)} + \frac{\gamma}{N} ~. \nonumber
\end{eqnarray}

The continuous time process is obtained by taking the limit
$\gamma \rightarrow 0$, i.e., the rate of change in $p_{ij}$ with
respect to \rev{$\gamma$} as $\gamma \rightarrow 0$. Then, dropping the
discrete time script $t$,
\begin{eqnarray}
\dot{p_{ij}}&=&\lim_{\gamma \rightarrow 0} \frac{d p_{ij}}{d\gamma} \nonumber \\
&=& \lim_{\gamma \rightarrow 0} \frac{ \left(\frac{-1}{N} A_{i,j}^{\gamma,t} + \left(p_{ij}-\frac{\gamma}{N}\right) \frac{U_{ij}}{p_{ij}N} A_{i,j}^{\gamma,t}\right) \left( 1 + \frac{p_{ij}-\frac{\gamma}{N}}{1-\gamma} \left( A_{i,j}^{\gamma,t}-1 \right) \right)}{ \left( 1 + \frac{p_{ij}-\frac{\gamma}{N}}{1-\gamma} \left( A_{i,j}^{\gamma,t}-1 \right) \right)^2 } \nonumber \\
&+& \frac{ \left(p_{ij}-\frac{\gamma}{N}\right) A_{i,j}^{\gamma,t} \left( \frac{p_{ij}-\frac{1}{N}}{(1-\gamma)^2} A_{i,j}^{\gamma,t} + \frac{p_{ij}-\frac{1}{N}}{1-\gamma} \left( \frac{\gamma}{N} A_{i,j}^{\gamma,t} \right) \right) } { \left( 1 + \frac{p_{ij}-\frac{\gamma}{N}}{1-\gamma} \left( A_{i,j}^{\gamma,t}-1 \right) \right)^2} + \frac{1}{N} \nonumber \\
&=& \frac{U_{ij}(1-p_{ij})}{N} ~. \label{eqn:exp3}
\end{eqnarray}

Consider the case $\sigma_i(t)= k \neq j$. Then,
\begin{eqnarray}
p_{ij}(t+1)&=& \frac{(1-\gamma) w_{ij}(t)} {\sum_{l=1}^N w_{il}(t) + w_{ik}(t) \left( A_{i,k}^{\gamma,t}-1 \right)} + \frac{\gamma}{N} \nonumber \\
&=& \frac{p_{ij}(t)-\frac{\gamma}{N}}{1 + \frac{p_{ik}(t)-\frac{\gamma}{N}}{1-\gamma} \left( A_{i,k}^{\gamma,t}-1 \right)} + \frac{\gamma}{N}. \nonumber
\end{eqnarray}
Thus
\begin{eqnarray}
\dot{p_{ij}}&=& \lim_{\gamma \rightarrow 0} \frac{ \frac{-1}{N} \left(1+ \frac{p_{ik}-\frac{\gamma}{N}}{1-\gamma}  \left( A_{i,k}^{\gamma,t}-1 \right)  \right)}{\left(1+ \frac{p_{ik}-\frac{\gamma}{N}}{1-\gamma}  \left( A_{i,k}^{\gamma,t}-1 \right) \right)^2} \nonumber \\
&+& \frac{ \left(p_{ij}-\frac{\gamma}{N}\right) \left( \frac{p_{ik}-\frac{1}{N}}{(1-\gamma)^2} A_{i,k}^{\gamma,t} + \frac{p_{ik}-\frac{1}{N}}{1-\gamma} \left( \frac{\gamma}{N} A_{i,k}^{\gamma,t} \right) \right)}{\left(1+ \frac{p_{ik}-\frac{\gamma}{N}}{1-\gamma}  \left( A_{i,k}^{\gamma,t}-1 \right) \right)^2} +\frac{1}{N} \nonumber \\
&=& -\frac{p_{ik}U_{ik}}{N} ~. \label{eqn:exp4}
\end{eqnarray}

Then from (\ref{eqn:exp3}) and (\ref{eqn:exp4}), the expected change in $p_{ij}$ with respect to the probability distribution $p_i$ of user $i$ over the channels is
\begin{eqnarray}
\bar{p}_{ij} = E_i[\dot{p}_{ij}]= \frac{1}{N} p_{ij} \sum_{l \in {\cal N}-\left\{j\right\}} p_{il} (U_{ij} - U_{il}). \nonumber
\end{eqnarray}
Taking the expectation with respect to the randomization of channel rates and other \rev{users'} actions we have
\begin{eqnarray}
\xi_{ij} &=& E_j[E_{-i}[\bar{p}_{ij}]] \nonumber \\
&=& \frac{1}{N} p_{ij} \sum_{l \in {\cal N}-\left\{j\right\}} p_{il} \left(E_j[ E_{-i}[ U_{ij}]] - E_j[E_{-i}[ U_{il}]] \right) \nonumber \\
&=&\frac{1}{N} (\mu_j p_{ij}) \sum_{l=1}^N p_{il} (\bar{g}_{ij}-\bar{g}_{il}) ~. \nonumber
\end{eqnarray} \qed
\end{proof}
Lemma \ref{lemma:exp1} shows that the dynamics of a user's
probability distribution over the actions is \rev{given by} a replicator equation
which is commonly studied in evolutionary game theory
\cite{sandholm1,smith1}. With this lemma we can establish the following theorem.

\begin{theorem} \label{thm:exp2}
For all but a measure zero subset of $[0,1]^{2N}$ from which \rev{the $\mu_j$'s and $g_j$'s} are selected, when $\gamma$ in Exp3 is arbitrarily small, the action profile converges to the set of PNE of the
congestion game $({\cal M}, {\cal N}, ({\cal S}_i)_{i\in {\cal M}},
(\mu_j g_j)_{j\in {\cal N}})$.
\end{theorem}
\begin{proof}
\rev{Because the replicator equation in Lemma \ref{lemma:exp1} is identical to the replicator equation in \cite{kleinberg1}, the proof of converge to PNE follows from \cite{kleinberg1}. Here, we briefly explain the steps in the proof.}
Defining the expected potential function to be
the expected value of the potential function $\phi$ where
expectation is taken with respect to the user's randomization one
can show that the solutions of the replicator equation converges
to the set of fixed points. Then the stability analysis using the
Jacobian matrix yields that every stable fixed point corresponds
to a Nash equilibrium. Then one can prove that for any stable
fixed point the eigenvalues of the Jacobian must be zero. This
implies that every stable fixed point corresponds to a {\em weakly
stable} Nash equilibrium strategy in the game theoretic sense.
Then using tools from algebraic geometry one can show that almost every
weakly stable Nash equilibrium is a pure Nash equilibrium of the
congestion game.

We also need to investigate the error introduced by treating the
discrete time update rule as a continuous time process. However,
by taking $\gamma$ infinitesimal we can approximate the discrete
time process by the continuous time process. For a discussion when
$\gamma$ is not infinitesimal one can define {\em approximately
stable equilibria} \cite{kleinberg1}. \qed
\end{proof}

The main difference between Exp3 and Hedge \cite{kleinberg1} is
that in Exp3 users do not need to observe the payoffs from the
channels that they do not select, whereas Hedge assumes complete
observation. In addition to that, we considered the dynamic channel
rates which is not considered in \cite{kleinberg1}.

\section{An Algorithm for Socially Optimal Allocation with Sub-linear Regret (Case 2)} \label{sec:case2}

In this section we propose an algorithm whose regret with respect to the socially optimal allocation is $O(n^{\frac{2M-1+2\gamma}{2M}})$ for $\gamma>0$ arbitrarily small. Clearly this regret is sublinear and approaches linear as the number of users $M$ increases. This means that the time average of the sum of the utilities of the players converges to the socially optimal welfare. Let ${\cal K} = \{k=(k_1, k_2, \ldots, k_N): k_j \geq 0, \forall j \in {\cal N}, k_1+k_2+\ldots+k_N=M \}$ denote an allocation of $M$ users to $N$ channels. Note that this allocation gives only the number of users on each channel. It does not say anything about which user uses which channel. We assume that the socially optimal allocation is unique up to permutations so $k^* = \arg\max_{k \in {\cal K}} \sum_{j=1}^N \mu_j k_j g_j(k_j)$ is unique. We also assume the following stability condition of the socially optimal allocation. Let $v_j(k_j)= \mu_j g_j(k_j)$. Then the stability condition says that $\arg\max_{k \in {\cal K}} \sum_{j=1}^N k_j \hat{v}_j(k_j)= k^*$ if $|\hat{v}_j(k)-v_j(k)| \leq \epsilon, \forall k \in \{1,2,\ldots,M\}, \forall j \in {\cal N}$, for some $\epsilon>0$, where $\hat{v}_j: \mathbb{N} \rightarrow \mathbb{R}$ is an arbitrary function. Let $T^i_{j,k}(t)$ be the number of times user $i$ used channel $j$ and observed $k$ users on it up to time $t$. We refer to the tuple $(j,k)$ as an arm. Let $n^i_{j,k}(t)$ be the time of the $t$th observation of user $i$ from arm $(j,k)$. Let $u^i_{j,k}(t)$ be the sample mean of the rewards from arm $(j,k)$ seen by user $i$ at the end of the $t$th play of arm $(j,k)$ by user $i$, i.e., $u^i_{j,k}(t)=(h_{j,k}(n^i_{j,k}(1)) + \ldots + h_{j,k}(n^i_{j,k}(t)))/t$. Then the socially optimal allocation estimated by user $i$ at time $t$ is $k^{i*}(t) = \arg\max_{k \in {\cal K}} \sum_{j=1}^N k_j u^i_{j,k}(t)$. The pseudocode of the Randomized Learning Algorithm (RLA) is given in Fig. \ref{fig:rla}. At time $t$ RLA explores with probability $1/(t^{\frac{1}{2M}-\frac{\gamma}{M}})$ by randomly choosing one of the channels and exploits with probability $1 - 1/(t^{\frac{1}{2M}-\frac{\gamma}{M}})$ by choosing a channel which is occupied by a user in the estimated socially optimal allocation.

The following will be useful in the proof of the main theorem of this section.
\begin{lemma} \label{lemma:case22}
Let $X_i, i=1,2,\ldots$ be a sequence of independent Bernoulli random variables such that $X_i$ has mean $q_i$ with $0 \leq q_i \leq 1$. Let $\bar{X}_k= \frac{1}{k} \sum_{i=1}^k X_i$ , $\bar{q}_k = \frac{1}{k} \sum_{i=1}^k q_i$. Then for any constant $\epsilon \geq 0$ and any integer $n \geq 0$,
\begin{eqnarray}
P\left( \bar{X}_n - \bar{q}_n \leq -\epsilon \right) \leq e^{-2n\epsilon^2}. \label{eqn:case26}
\end{eqnarray} 
\end{lemma}
\begin{proof}
The result follows from symmetry and \cite{turner1}. 
\qed
\end{proof}

\begin{lemma} \label{lemma:case23}
For $p>0, p \neq 1$
\begin{eqnarray}
\frac{(n+1)^{1-p}-1}{1-p} <	\sum_{t=1}^n \frac{1}{t^p} < 1 + \frac{n^{1-p}-1}{1-p} \label{eqn:case2_7}
\end{eqnarray}
\end{lemma}
\begin{proof}
See \cite{chlebus1}. \qed
\end{proof}

\begin{figure}[htb]
\fbox {
\begin{minipage}{\columnwidth}
RLA (for user $i$)
\begin{algorithmic}[1]
\STATE{Initialize: $0<\gamma << 1$, $u^i_{j,k}(1)=0, T^i_{j,k}(1)=0, \forall j \in {\cal N}, k \in {\cal M}$, $t=1$, sample $\sigma_i(1)$ uniformly from ${\cal N}$.}
\WHILE{$t>0$}
\STATE{play channel $\sigma_i(t)$, observe $l(t)$ the total number of players using channel $\sigma_i(t)$ and reward $h_{\sigma_i(t), l(t)}(t)$.}
\STATE{Set $T^i_{\sigma_i(t),l(t)}(t+1)=T^i_{\sigma_i(t),l(t)}(t)+1$.}
\STATE{Set $T^i_{j,l}(t+1)=T^i_{j,l}(t)$ for $(j,l) \neq (\sigma_i(t),l(t))$.}
\STATE{Set $u^i_{\sigma_i(t),l(t)}(t+1)= \frac{T^i_{\sigma_i(t),l(t)}(t) u^i_{\sigma_i(t),l(t)}(t) + h_{\sigma_i(t), l(t)}(t)}{T^i_{\sigma_i(t),l(t)}(t+1)}$.}
\STATE{Set $u^i_{j,l}(t+1) = u^i_{j,l}(t)$ for $(j,l) \neq (\sigma_i(t),l(t))$.}
\STATE{Set $k^{i*}(t+1) = \arg\max_{k \in {\cal K}} \sum_{j=1}^N k_j u^i_{j,k_j}(t+1)$.}
\STATE{Set $\theta^{*i}(t+1)$ to be the set of channels used by at least one user in $k^{*i}(t+1)$.}
\STATE{Draw $i_t$ randomly from Bernoulli distribution with $P(i_t=1)=\frac{1}{t^{(1/2M) - \gamma/M}}$}
\IF{$i_t=0$} 
\IF{$\sigma_i(t) \in \theta^*(t+1)$ and $l(t)=k^{i*}_j(t+1)$}
\STATE{$\sigma_i(t+1) = \sigma_i(t)$}
\ELSE
\STATE{$\sigma_i(t+1)$ is selected uniformly at random from the channels in $\theta^*(t+1)$.}
\ENDIF
\ELSE
\STATE{Draw $\sigma_i(t+1)$ uniformly at random from ${\cal N}$.}
\ENDIF
\STATE{$t=t+1$}
\ENDWHILE
\end{algorithmic}
\end{minipage}}
\caption{pseudocode of RLA} \label{fig:rla}
\end{figure}

\begin{theorem} \label{thm:case2}
When all players use RLA the regret with respect to the socially optimal allocation is $O(n^{\frac{2M-1+2\gamma}{2M}})$ where $\gamma$ can be arbitrarily small.
\end{theorem}
\begin{proof}  
Let $H(t)$ be the event that at time $t$ there exists at least one user that computed the socially optimal allocation \rev{incorrectly}.  Let $\omega$ be a sample path. Then
\begin{eqnarray}
&& \sum_{t=1}^n I(\omega \in H(t)) \leq \sum_{t=1}^n \sum_{i=1}^M I(k^{*i}(t) \neq k^*) \nonumber \\
&\leq& \sum_{(t,i,j,l)=(1,1,1,1)}^{(n, M, N, M)} I(|u^i_{j,l}(T^i_{j,l}(t))-v_j(l)| \geq \epsilon) \nonumber
\end{eqnarray}
\begin{eqnarray}
&=& \sum_{(t,i,j,l)=(1,1,1,1)}^{(n, M, N, M)} I\left(|u^i_{j,l}(T^i_{j,l}(t))-v_j(l)| \geq \epsilon, T^i_{j,l}(t) \geq \frac{a \ln t}{\epsilon^2}\right) \nonumber \\
&+& \sum_{(t,i,j,l)=(1,1,1,1)}^{(n, M, N, M)} I\left(|u^i_{j,l}(T^i_{j,l}(t))-v_j(l)| \geq \epsilon, T^i_{j,l}(t) < \frac{a \ln t}{\epsilon^2}\right) \label{eqn:case21}
\end{eqnarray}
\com{can we fixed the equation numbering (8) in the above equation?} 

Let $\epsilon^i_{j,k}(t) = \sqrt{\frac{a \ln t}{T^i_{j,k}(t)}}$. Then $T^i_{j,k}(t) \geq \frac{a \ln t}{\epsilon^2} \Rightarrow \epsilon \geq \sqrt{\frac{a \ln t}{T^i_{j,k}(t)}} = \epsilon^i_{j,k}(t)$. Therefore, 
\begin{eqnarray}
I\left(|u^i_{j,l}(T^i_{j,l}(t))-v_j(l)| \geq \epsilon, T^i_{j,l}(t) \geq \frac{a \ln t}{\epsilon^2}\right) &\leq& I\left(|u^i_{j,l}(T^i_{j,l}(t))-v_j(l)| \geq \epsilon^i_{j,l}(t) \right) \nonumber \\
I\left(|u^i_{j,l}(T^i_{j,l}(t))-v_j(l)| \geq \epsilon, T^i_{j,l}(t) < \frac{a \ln t}{\epsilon^2}\right) &\leq& I\left( T^i_{j,l}(t) < \frac{a \ln t}{\epsilon^2}\right) \nonumber
\end{eqnarray}
Then, continuing from (\ref{eqn:case21}),
\begin{eqnarray}
&& \sum_{t=1}^n I(\omega \in H(t)) \nonumber \\
&\leq&  \sum_{(t,i,j,l)=(1,1,1,1)}^{(n, M, N, M)} \left( I\left(|u^i_{j,l}(T^i_{j,l}(t))-v_j(l)| \geq \epsilon^i_{j,l}(t) \right) + I\left( T^i_{j,l}(t) < \frac{a \ln t}{\epsilon^2}\right) \right) \label{eqn:case22}
\end{eqnarray}
Taking the expectation \rev{over} (\ref{eqn:case22}),
\begin{eqnarray}
&& E\left[ \sum_{t=1}^n I(\omega \in H(t)) \right] \nonumber \\
&\leq& \sum_{(t,i,j,l)=(1,1,1,1)}^{(n, M, N, M)}  P\left(|u^i_{j,l}(T^i_{j,l}(t))-v_j(l)| \geq \epsilon^i_{j,l}(t) \right) \nonumber \\
&+& \sum_{(t,i,j,l)=(1,1,1,1)}^{(n, M, N, M)} P\left( T^i_{j,l}(t) < \frac{a \ln t}{\epsilon^2} \right). \label{eqn:case23}
\end{eqnarray}
We have
\begin{eqnarray}
&&P\left(|u^i_{j,l}(T^i_{j,l}(t))-v_j(l)| \geq \epsilon^i_{j,l}(t) \right) \nonumber \\
&=& P\left(u^i_{j,l}(T^i_{j,l}(t))-v_j(l) \geq \epsilon^i_{j,l}(t) \right) + P\left(u^i_{j,l}(T^i_{j,l}(t))-v_j(l) \leq -\epsilon^i_{j,l}(t) \right) \nonumber \\
&=&P\left( \frac{S^i_{j,l}(T^i_{j,l}(t))}{T^i_{j,l}(t)} - v_j(l) \geq \epsilon^i_{j,l}(t) \right) + P\left( \frac{S^i_{j,l}(T^i_{j,l}(t))}{T^i_{j,l}(t)} - v_j(l) \leq - \epsilon^i_{j,l}(t) \right) \nonumber \\
%&=& P\left( S^i_{j,l}(T^i_{j,l}(t)) \geq T^i_{j,l}(t) v_j(l) + T^i_{j,l}(t) \epsilon^i_{j,l}(t) \right) + P\left( S^i_{j,l}(T^i_{j,l}(t)) \leq T^i_{j,l}(t) v_j(l) - T^i_{j,l}(t) \epsilon^i_{j,l}(t) \right)
%\nonumber \\
&\leq& 2 \exp\left(- \frac{2 (T^i_{j,l}(t))^2 (\epsilon^i_{j,l}(t))^2}{T^i_{j,l}(t)}\right) = 2 \exp\left(- \frac{2 T^i_{j,l}(t) a \ln t}{T^i_{j,l}(t)}\right) = \frac{2}{t^{2a}} ~, \label{eqn:case24}
\end{eqnarray}
\rev{where} (\ref{eqn:case24}) follows from the Chernoff-Hoeffding inequality. 

Now we will bound $P\left( T^i_{j,l}(t) < \frac{a \ln t}{\epsilon^2} \right)$. Let $TR^i_{j,l}(t)$ be the number of time steps in which player $i$ played channel $j$ and observed $l$ users on channel $j$ in the time steps where all players randomized up to time $t$. Then
\begin{eqnarray}
\{ \omega: T^i_{j,l}(t) < \frac{a \ln t}{\epsilon^2} \} \subset \{ \omega: TR^i_{j,l}(t) < \frac{a \ln t}{\epsilon^2} \}, \nonumber \\
\end{eqnarray}
Thus
\begin{eqnarray}
P\left( T^i_{j,l}(t) < \frac{a \ln t}{\epsilon^2} \right) \leq P\left( TR^i_{j,l}(t) < \frac{a \ln t}{\epsilon^2} \right) ~. \label{eqn:case2_8}
\end{eqnarray}
Now we define new Bernoulli random variables $X^i_{j,l}(s)$ as follows: $X^i_{j,l}(s)=1$ if all players randomize at time $s$ and player $i$ selects channel $j$ and observes $l$ players on it according to the random draw. $X^i_{j,l}(s)=0$ else. Then $TR^i_{j,l}(t)= \sum_{s=1}^t X^i_{j,l}(s)$. $P(X^i_{j,l}(s)=1)= \rho_s p_l$ where $p_l= \frac{{M-1 \choose l-1}{M+N-l-2 \choose N-2}}{ {M+N-1 \choose N-1}}$ and $\rho_s = \frac{1}{s^{(1/2)-\gamma}}$. Let $s_t = \sum_{s=1}^t \frac{1}{s^{(1/2)-\gamma}}$  Then
\begin{eqnarray}
&& P\left( TR^i_{j,l}(t) < \frac{a \ln t}{\epsilon^2} \right) \nonumber \\
&=& P\left( \frac{TR^i_{j,l}(t)}{t} - \frac{p_k s_t}{t} < \frac{a \ln t}{t \epsilon^2} - \frac{p_k s_t}{t} \right) \nonumber \\
&\leq& P\left( \frac{TR^i_{j,l}(t)}{t} - \frac{p_k s_t}{t} < \frac{a \ln t}{t \epsilon^2} - \frac{p_k (t+1)^{(1/2)+\gamma}-1}{t((1/2)+\gamma)} \right), \label{eqn:case29}
\end{eqnarray}
where (\ref{eqn:case29}) follows from Lemma \ref{lemma:case23}.
Let $\tau(M,N,\epsilon,\gamma,\gamma',a)$ be the time that for all $k \in \{1,2,\ldots,M \}$.
\begin{eqnarray}
\frac{p_k (t+1)^{(1/2)+\gamma}-1}{t((1/2)+\gamma)} - \frac{a \ln t}{t \epsilon^2} \geq t^{(1/2)+\gamma'}, \label{eqn:case210}
\end{eqnarray}
where $0 < \gamma' < \gamma$. Then for all $t \geq \tau(M,N,\epsilon,\gamma,\gamma',a)$ (\ref{eqn:case210}) will hold since RHS increases faster than LHS. Thus we have for $t \geq \tau(M,N,\epsilon,\gamma,\gamma',a)$
\begin{eqnarray}
&& P\left( \frac{TR^i_{j,l}(t)}{t} - \frac{p_k s_t}{t} < \frac{a \ln t}{t \epsilon^2} - \frac{p_k (t+1)^{(1/2)+\gamma}-1}{t((1/2)+\gamma)} \right) \nonumber \\
&\leq& P\left( \frac{TR^i_{j,l}(t)}{t} - \frac{p_k s_t}{t} < {t^{-(1/2)+\gamma'} } \right) \nonumber \\
&\leq& e^{-2 t t^{2\gamma'-1}} = e^{-2 t^{2\gamma'}} \leq e^{-2 \ln t} = \frac{1}{t^2}. \label{eqn:case211}
\end{eqnarray}
Let $a=1$. Then continuing from (\ref{eqn:case23}) by substituting (\ref{eqn:case24}) and (\ref{eqn:case211}) \rev{we have} 
\begin{eqnarray}
E\left[ \sum_{t=1}^n I(\omega \in H(t)) \right] \leq M^2 N \left(\tau(M,N,\epsilon,\gamma,\gamma',1) + 3 \sum_{t=1}^n \frac{1}{t^2} \right). \label{eqn:case212}
\end{eqnarray} 
Thus we proved that the expected number of time steps in which there exists at least one user that computed the socially optimal allocation incorrectly is finite. Note that because RLA explores with probability $\frac{1}{t^{1/2M-\gamma/M}}$, the expected number of time steps in which all the players are not randomizing up to time $n$ is 
\begin{eqnarray}
\sum_{t=1}^n \left( 1- \left(1-\frac{1}{t^{(1/2M)-\gamma/M}}\right)^M\right) \leq \sum_{t=1}^n \frac{M}{t^{1/2M - \gamma/M}} = O(n^{\frac{2M-1+2\gamma}{2M}}) . \label{eqn:case213}
\end{eqnarray}
Note that players can choose $\gamma$ arbitrarily small, increasing the finite regret due to $\tau(M,N,\epsilon,\gamma,\gamma',1)$. Thus if we are interested in the asymptotic performance then $\gamma>0$ can be arbitrarily small.

Now we do the worst case analysis. We classify the time steps into two. {\em Good} time steps in which all the players know the socially optimal allocation correctly and none of the players randomize excluding the randomizations done for settling down to the socially optimal allocation. {\em Bad} time steps in which there exists a player that does not know the socially optimal allocation correctly or there is a player that randomizes excluding the randomizations done for settling down to the socially optimal allocation. The number of {\em Bad} time steps in which there exists a player that does not know the socially optimal allocation correctly is finite while the number of time steps in which there is a player that randomizes excluding the randomizations done for settling down to the socially optimal allocation is $O(n^{\frac{2M-1+2\gamma}{2M}})$. The worst case is when each bad step is followed by a good step. Then from this good step the expected number of times to settle down to the socially optimal allocation is $\left(1-\frac{1}{{M+z^*-1 \choose z^*-1}}\right)/\left(\frac{1}{{M+z^*-1 \choose z^*-1}}\right)$ where $z^*$ is the number of channels which has at least one user in the socially optimal allocation. Assuming in the worst case the sum of the utilities of the players is $0$ when they are not playing the socially optimal allocation we have 
\begin{eqnarray}
R(n) &\leq& \frac{1-\frac{1}{{M+z^*-1 \choose z^*-1}}}{\frac{1}{{M+z^*-1 \choose z^*-1}}} \left( M^2 N \left(\tau(M,N,\epsilon,\gamma,\gamma',1) + 3 \sum_{t=1}^n \frac{1}{t^2} \right) + O(n^{\frac{2M-1+2\gamma}{2M}}) \right) \nonumber \\
&=& O(n^{\frac{2M-1+2\gamma}{2M}}) \nonumber
\end{eqnarray}
\qed
\end{proof}

Note that we mentioned earlier, under a classical multi-armed bandit problem approach as cited before \cite{anandkumar,anantharam2,anantharam1,lai1,liu2,tekin,tekin2},  a logarithmic regret $O(\log n)$ is achievable.  The  fundamental difference between these studies and the problem in the present paper is the following: Assume that at time $t$ user $i$ selects channel $j$. This means that $i$ selects to observe an arm from the set $\{(j,k): k \in {\cal M}\}$ but the arm assigned to $i$ is selected from this set depending on the choices of other players. 

Also note that in RLA a user computes the socially optimal allocation according to its estimates at each time step. This could pose significant computational effort since integer programming is NP-hard in general. However, by exploiting the stability condition on the socially optimal allocation a user may reduce the number of computations; this is a subject of future research.

\section{An Algorithm for Socially Optimal Allocation (Case 3)} \label{sec:case3}

In this section we assume that $g_j(n)$ is decreasing in $n$ for
all $j \in {\cal N}$. For simplicity we assume that the socially
optimal allocation is unique up to the permutations of $\sigma^*$.
When this uniqueness assumption does not hold we need a more
complicated algorithm to achieve the socially optimal allocation.
All users use the Random Selection (RS) algorithm defined in
Fig. \ref{fig:alg1}. RS consists of two phases. Phase 1 is the
learning phase where the user randomizes to learn the interference
functions. Let $B_j(t)$ be the set of distinct payoffs observed
from channel $j$ up to time $t$. Then the payoffs in set $B_j(t)$
can be ordered in a decreasing way with the associated indices
$\{1,2,\ldots,|B_j(t)|\}$. Let $O(B_j(t))$ denote this ordering.
Since the IFs are decreasing, at the time $|B_j(t)|=M$, the user
has learned $g_j$. At the time $|\cup_{j=1}^N B_j(t)|=MN$, the
user has learned all IFs. Then, the user computes $\mathcal{A}^*$
and phase 2 of RS starts where the user randomizes to converge to
the socially optimal allocation.

\begin{figure}[htb]
\fbox {
\begin{minipage}{\columnwidth}
{Random Selection (RS)}
\begin{algorithmic}[1]
\STATE{Initialize: $t=1$, $b=0$, $B_j(1)=\emptyset, \forall j \in
{\cal N}$, sample $\sigma_i(1)$ from the uniform distribution on
$\mathcal{N}$} \STATE{Phase 1} \WHILE{$b<MN$}
\IF{$h_{\sigma_i(t)}(t) \notin B_{\sigma_i(t)}(t)$}
\STATE{$B_{\sigma_i(t+1)}(t+1) \leftarrow O(B_{\sigma_i(t)}(t)
\cup h_{\sigma_i(t)}(t))$} \STATE{$b=b+1$} \ENDIF \STATE{Sample
$\sigma_i(t+1)$ from the uniform distribution on $\mathcal{N}$}
\STATE{$t=t+1$} \ENDWHILE  \STATE{find the socially
optimal allocation $\sigma^*$} \STATE{Phase 2} \WHILE{$b\geq MN$}
\IF{$h_{\sigma_i(t)}(t) < v^*_{\sigma_i(t)}$}
\STATE{Sample $\sigma_i(t+1)$ from the uniform distribution on
$\mathcal{N}$} \ELSE \STATE{$\sigma_i(t+1) = \sigma_i(t)$} \ENDIF
\STATE{$t=t+1$} \ENDWHILE
\end{algorithmic}
\end{minipage}
} \caption{pseudocode of RS} \label{fig:alg1}
\end{figure}

\begin{theorem} \label{thm:theorem1}
Under the assumptions of (C3) if all players use RS algorithm to
choose their actions, then the expected time to converge to the
socially optimal allocation is finite.
\end{theorem}
\begin{proof}
Let $T_{OPT}$ denote the time the socially optimal allocation is
achieved, $T_L$ be the time when all users learn all the IFs, $T_F$ be the time it takes to reach the socially
optimal allocation after all users learn all the IFs. Then $T_{OPT}=T_L+T_F$ and $E[T_{OPT}]=E[T_L]+E[T_F]$.
We will bound $E[T_L]$ and $E[T_F]$. Let $T_{i}$ be the first time
that $i$ users have learned the IFs. Let
$\tau_i=T_i-T_{i-1}, i=1,2,\ldots,M$ and $T_0=0$. Then
$T_L=\tau_1+\ldots+\tau_M$. Define a Markov chain over all $N^M$
possible configurations of $M$ users over $N$ channels based on
the randomization of the algorithm. This Markov chain has a time
dependent stochastic matrix which changes at times $T_1, T_2,
\ldots, T_{M}$. Let $P_{T_0}, P_{T_1}, \ldots, P_{T_M}$ denote the
stochastic matrices after the times $T_0, T_1, \ldots, T_{M}$
respectively. This Markov chain is irreducible at all times up to
$T_M$ and is reducible with absorbing states corresponding to the
socially optimal allocations after $T_M$. Let $\hat{T}_1,
\hat{T}_2, \ldots \hat{T}_{M}$ be the times that all
configurations are visited when the Markov chain has stochastic
matrices $P_{T_0}, P_{T_1}, \ldots, P_{T_{M-1}}$ respectively.
Then because of irreducibility and finite
states $E[\hat{T}_i]<z_1, i=1,\ldots,M$ for some constant
$z_1>0$ . Since $\tau_i \leq \hat{T}_i, i=1,\ldots,M$ a.s. we have
$E[T_L]<Mz_1$. For the Markov chain with stochastic matrix
$P_{T_M}$ all the configurations that do not correspond to the
socially optimal allocation are transient states. Since starting
from any transient state the mean time to absorption is finite
$E[T_F]<z_2$, for some constant $z_2>0$. \qed
\end{proof}

\section{Conclusion} \label{sec:conc}

In this paper we studied the decentralized multiuser resource allocation problem with various levels of communication and cooperation between the users. Under three different scenarios we proposed three algorithms with reasonable performance. Our future reserach will include characterization of achievable performance regions for these scenarios. For example, in case 2 we are interested in finding an optimal algorithm and a lower bound on the performance. 

\bibliographystyle{splncs03}
\bibliography{OSA}

\begin{thebibliography}{10}
\providecommand{\url}[1]{\texttt{#1}}
\providecommand{\urlprefix}{URL }

\bibitem{agrawal1}
Agrawal, R.: {Sample Mean Based Index Policies with $O(\log(n))$ Regret for the
  Multi-armed Bandit Problem}. Advances in Applied Probability  27(4),
  1054--1078 (December 1995)

\bibitem{ahmad2}
Ahmad, S., Tekin, C., Liu, M., Southwell, R., Huang, J.: {Spectrum Sharing as
  Spatial Congestion Games}. http://arxiv.org/abs/1011.5384  (2010)

\bibitem{anandkumar}
Anandkumar, A., Michael, N., Tang, A.: {Opportunistic Spectrum Access with
  Multiple Players: Learning under Competition}. In: Proc. of IEEE INFOCOM
  (March 2010)

\bibitem{anantharam2}
Anantharam, V., Varaiya, P., Walrand, J.: {Asymptotically Efficient Allocation
  Rules for the Multiarmed Bandit Problem with Multiple Plays-Part I: IID
  Rewards}. IEEE Trans. Automat. Contr. pp. 968--975 (November 1987)

\bibitem{anantharam1}
Anantharam, V., Varaiya, P., Walrand, J.: {Asymptotically Efficient Allocation
  Rules for the Multiarmed Bandit Problem with Multiple Plays-Part II:
  Markovian Rewards}. IEEE Trans. Automat. Contr. pp. 977--982 (November 1987)

\bibitem{auer}
Auer, P., Cesa-Bianchi, N., Fischer, P.: {Finite-time Analysis of the
  Multiarmed Bandit Problem}. Machine Learning  47,  235--256 (2002)

\bibitem{auer2}
Auer, P., Cesa-Bianchi, N., Freund, Y., Schapire, R.: {The Nonstochastic
  Multiarmed Bandit Problem}. SIAM Journal on Computing  32,  48--77 (2002)

\bibitem{chlebus1}
Chlebus, E.: {An Approximate Formula for a Partial Sum of the Divergent
  p-series}. Applied Mathematics Letters  22,  732--737 (2009)

\bibitem{turner1}
D.W.~Turner, D.M.~Young, J.S.: {A Kolmogorov Inequality for the Sum of
  Independent Bernoulli Random Variables with Unequal Means}. Statistics and
  Probability Letters  23,  243--245 (1995)

\bibitem{freund1}
Freund, Y., Schapire, R.: {Adaptive Game Playing Using Multiplicative Weights}.
  Games and Economic Behaviour  29,  79--103 (1999)

\bibitem{krishnamachari}
Gai, Y., Krishnamachari, B., Jain, R.: {Learning Multiuser Channel Allocations
  in Cognitive Radio Networks: a Combinatorial Multi-armed Bandit Formulation}.
  In: IEEE Symp. on Dynamic Spectrum Access Networks (DySPAN) (April 2010)

\bibitem{kakhbod1}
Kakhbod, A., Teneketzis, D.: {Power Allocation and Spectrum Sharing in
  Cognitive Radio Networks With Strategic Users}. In: 49th IEEE Conference on
  Decision and Control (CDC) (December 2010)

\bibitem{gaurav2010}
Kasbekar, G., Proutiere, A.: {Opportunustic Medium Access in Multi-channel
  Wireless Systems: A Learning Approach}. In: Proceedings of the 48th Annual
  Allerton Conference on Communication, Control, and Computation (September
  2010)

\bibitem{kleinberg1}
Kleinberg, R., Piliouras, G., Tardos, E.: {Multiplicative Updates Outperform
  Generic No-Regret Learning in Congestion Games}. In: Annual ACM Symposium on
  Theory of Computing (STOC) (2009)

\bibitem{lai1}
Lai, T., Robbins, H.: {Asymptotically Efficient Adaptive Allocation Rules}.
  Advances in Applied Mathematics  6,  4--22 (1985)

\bibitem{liu2}
Liu, K., Zhao, Q.: {Distributed Learning in Multi-Armed Bandit with Multiple
  Players}. IEEE Transactions on Signal Processing  58(11),  5667--5681
  (November 2010)

\bibitem{monderer1}
Monderer, D., Shapley, L.S.: {Potential Games}. Games and Economic Behavior
  14(1),  124--143 (1996)

\bibitem{rosenthal1}
Rosenthal, R.: {A Class of Games Possessing Pure-strategy Nash Equilibria}.
  International Journal of Game Theory  2,  65--67 (1973)

\bibitem{sandholm1}
Sandholm, W.H.: {Population Games and Evolutionary Dynamics}. Manuscript
  (2008)

\bibitem{smith1}
Smith, J.M.: {Evolution and the Theory of Games}. Cambridge University Press
  (1982)

\bibitem{tekin}
Tekin, C., Liu, M.: {Online Algorithms for the Multi-armed Bandit Problem with
  Markovian Rewards}. In: Proceedings of the 48th Annual Allerton Conference on
  Communication, Control, and Computation (September 2010)

\bibitem{tekin2}
Tekin, C., Liu, M.: {Online Learning in Opportunistic Spectrum Access: A
  Restless Bandit Approach}. In: 30th IEEE International Conference on Computer
  Communications (INFOCOM) (April 2011)

\end{thebibliography}

\end{document}